\documentclass[10pt,final,journal,twoside,letterpaper]{IEEEtran}
\usepackage{graphicx,amsmath,amssymb,verbatim}
\usepackage[noadjust]{cite}
\usepackage{array}
\usepackage{cases}
\newcommand{\rk}[1]{\textup{\mbox{\,rank}}\,(#1)}
\newtheorem{theorem}{Theorem}[section]
\newtheorem{proposition}[theorem]{Proposition}

\begin{document}
\title{A characterization of entanglement-assisted quantum low-density parity-check codes}
\author{Yuichiro~Fujiwara,~\IEEEmembership{Member,~IEEE,} and Vladimir~D.~Tonchev%
\thanks{This work is supported by JSPS Grants-in-Aid for Scientific Research 20$\cdot$5897 and JSPS Postdoctoral Fellowships for Research Abroad (YF),
and by NSA Grant H98230-10-1-0177 and H98230-12-0213 (VDT).
The material in this paper was presented in part
at the Second International Conference on Quantum Error Correction, Los Angeles, USA, December 2011.}%
\thanks{Y. Fujiwara is with the Division of Physics, Mathematics and Astronomy, California Institute of Technology, MC 253-37, Pasadena, CA 91125 USA
{(email: yuichiro.fujiwara@caltech.edu)}.}
\thanks{V. D. Tonchev is with the Department of Mathematical Sciences, Michigan Technological University, Houghton, MI 49931 USA
{(email: tonchev@mtu.edu)}.}
\thanks{Copyright \copyright\ 2012 IEEE. Personal use of this material is permitted.
However, permission to use this material for any other purposes must be obtained from the IEEE by sending a request to pubs-permissions@ieee.org.}}
\markboth{IEEE transactions on Information Theory,~Vol.~x, No.~xx,~month~year}
{Fujiwara and Tonchev: Characterization of entanglement-assisted quantum LDPC codes}

\maketitle

\begin{abstract}
As in classical coding theory, quantum analogues of low-density parity-check (LDPC) codes
have offered good error correction performance
and low decoding complexity by employing the Calderbank-Shor-Steane (CSS) construction.
However, special requirements in the quantum setting severely limit the structures such 
quantum codes can have.
While the entanglement-assisted stabilizer formalism overcomes this limitation
by exploiting maximally entangled states (ebits),
excessive reliance on ebits is a substantial obstacle to implementation.
This paper gives necessary and sufficient conditions for the existence of quantum LDPC codes
which are obtainable from pairs of identical LDPC codes and consume only one ebit,
and studies the spectrum of attainable code parameters.
\end{abstract}

\begin{keywords}
Entanglement-assisted quantum error correction, low-density parity-check code, stabilizer code,
Steiner $2$-design, pairwise balanced design.
\end{keywords}

\IEEEpeerreviewmaketitle

\section{Introduction}
\PARstart{T}{his} paper addresses the question of how much one pair of qubits in a maximally entangled state can be exploited
to import classical sparse graph codes by the entanglement-assisted stabilizer formalism
proposed by Brun, Devetak, and Hsieh \cite{BDH}.
We will show how quantum error-correcting codes with particular desirable properties
under this framework are equivalent to some fundamental objects from combinatorial design theory,
known as pairwise balanced designs \cite{BJL}.
While  earlier relevant results in the literature give sufficient conditions
for the existence of entanglement-assisted quantum codes based on classical sparse graph codes
(see \cite{HYH2,FCVBT}, and references therein),
results presented here give necessary and sufficient conditions
under the conventional standard assumptions,
and mathematically describe the structure of such quantum codes consuming only one ebit.

Low-density parity-check (LDPC) codes are among the best known classical codes
in terms of error correction performance and decoding complexity \cite{RUbook}.
Extensive efforts have been made to generalize this class of error-correcting schemes
in classical coding theory to the quantum setting.
Among others, codes obtained by applying the well-known Calderbank-Shor-Steane (CSS) construction \cite{CS, Steane}
to pairs of carefully chosen identical LDPC codes have shown remarkable error correcting performance in simulations
(see, for example, \cite{MMM,Aly,Djordjevic,HBD} for recently proposed combinatorial quantum LDPC codes).

However, the progress on the quantum analogues of LDPC codes has lagged behind their 
classical counterparts
because of special requirements imposed on the code structure in the quantum setting;
only a limited class of classical codes can be exploited in a direct manner.

The development of the entanglement-assisted stabilizer formalism is a recent breakthrough in this regard
\cite{BDH}, \cite{HDB}.
This framework allows us to import any binary or quaternary linear codes
to the quantum setting by exploiting maximally entangled states.
In other words, by taking advantage of ebits,
the code designer can turn good classical linear codes into quantum error-correcting codes
and expect similar good performance in the quantum setting.
In fact, Hsieh, Brun, and Devetak \cite{HBD} demonstrated this advantage
by constructing entanglement-assisted quantum LDPC codes
which have notable error correction performance.

However, the entanglement-assisted stabilizer formalism is not a silver bullet;
an adequate supply of ebits may not be available.
The number of ebits required to import a given pair of classical LDPC codes varies greatly from pair to pair.
In fact, as stated by Hsieh, Yen, and Hsu \cite{HYH2}, until recently it was conjectured that
the entanglement-assisted stabilizer formalism required an impractically large number of ebits
to make use of classical codes with good error correction performance.

Fortunately, Hsieh, Yen, and Hsu \cite{HYH2} and Fujiwara, Clark, Vandendriessche, De Boeck, and Tonchev
 \cite{FCVBT} independently 
disproved this conjecture by giving examples of codes which require only one ebit
while outperforming the previously known quantum LDPC codes in simulations.
Given the positive results on quantum error-correcting codes requiring a tiny amount of entanglement
and the fact that excessive reliance on ebits is a substantial obstacle to implementation,
it is of interest to investigate the characteristics of quantum codes consuming only one ebit.

In this paper, we investigate what kind of quantum LDPC code is obtainable 
if only one ebit is allowed to import pairs of identical classical LDPC codes.
We show the equivalence between such quantum LDPC codes
and special classes of combinatorial objects,
namely Steiner $2$-designs and pairwise balanced designs of index one.
(For a thorough introduction to combinatorial design theory, we refer the interested reader to \cite{BJL,TRIPLESYSTEMS}.)
This equivalence provides theoretical insight into the properties and attainable code parameters,
and explains why all known high performance quantum LDPC codes requiring only one ebit
were derived from combinatorial designs of this kind.

It should be noted that it is also possible to utilize a pair of nonidentical classical error-correcting codes to construct a quantum LDPC code
as long as they are of the same length.
In fact, very recently quantum LDPC codes with good error correction performance have been found
through a clever use of nonidentical ingredients \cite{KHIS,HKIS}.
However, their methods require a large number of physical qubits to encode,
which is at odds with the focus of the current paper, that is, shedding light on more easily implementable quantum LDPC codes.
For this reason, we leave the equally interesting case of nonidentical ingredients consuming only a small amount of entanglement to future work.

In the following sections, we will show how the requirement of consuming only one ebit
dictates the structure of the exploitable pairs of identical classical LDPC codes.
In Section \ref{sc:EAQLDPC}, we briefly review the entanglement-assisted quantum LDPC codes
and related facts from combinatorial design theory and then prove the
equivalence between pairwise balanced designs of index one and entanglement-assisted
quantum LDPC codes consuming one ebit.
Section \ref{sc:Rates} provides bounds for the code parameters.
Section \ref{sc:Conclusion} discusses some open questions and directions for future work.

\section{Entanglement-assisted quantum \textup{LDPC} codes}\label{sc:EAQLDPC}
In this section we study the existence of entanglement-assisted quantum LDPC codes consuming one ebit
and its relation to combinatorial designs.
For a concise introduction to the entanglement-assisted stabilizer formalism,
we refer the reader to Hsieh, Yen, and Hsu \cite{HYH2}.

An $[[n,k;c]]$ \textit{entanglement-assisted quantum error-correcting code} 
(EAQECC) encodes $k$ logical qubits
into $n$ physical qubits with the help of $c$ copies of maximally entangled states.
The parameters $n$ and $k$ are the \textit{length} and \textit{dimension} of the 
EAQECC respectively.
An  $[[n,k;c]]$  EAQECC requires $c$ \textit{ebits}.

A classical LDPC code is a binary linear code which admits a parity-check matrix
with a small number of nonzero entries.
The \textit{quantum check matrix} of an EAQECC of length $n$
constructed by the CSS construction with a pair of LDPC codes has the form
\[\left[\begin{array}{cc}
H_1 & 0 \\
0 & H_2
\end{array}\right],\]
where $H_1$ and $H_2$ are the parity-check matrices of classical LDPC codes of length $n$,
and $0$ represents a zero matrix.
The EAQECC is called \textit{homogeneous} if $H_1$ is obtained by permuting rows of $H_2$.
In this case, because permuting rows does not affect the code parameters,
without loss of generality we assume $H_1 = H_2$ and omit the subscripts.
As far as the authors are aware, at the time of writing,
all entanglement-assisted quantum LDPC codes proposed in the literature
for the depolarizing channel are homogeneous.

The required amount of entanglement can be calculated
from the rank of $H_1H_2^T$ over the field of order $2$, that is, its $2$-rank
(see Wilde and Brun \cite{WB} for the proof and Wilde \cite{Wilde} for an alternative, equivalent formula).
Since we do not use ranks over another field,
we always assume that ranks are computed over $\mathbb{F}_2$.
If the parity-check matrix $H$ defines an $[n,k,d]$ linear code $C$,
then the resulting homogeneous code requires $c = \rk{HH^T}$ ebits
and is of length $n$ and dimension $2k - n + c$ \cite{HDB}.
The case when $c = 0$ gives the well-known \textit{stabilizer code}
of \textit{minimum distance} $d'$, where $d' \geq d$ is the minimum Hamming weight of a codeword in $C \setminus C^{\perp}$
(see \cite{MikeandIke} for a more detailed treatment of this special case).
Similarly, the minimum distance for the case $c >  0$ is
the minimum Hamming weight of a codeword in $C \setminus \mathcal{R}(C^{\perp})$,
where $\mathcal{R}(C^{\perp})$ is the normal subgroup $C \cap C^{\perp}$.
A particularly useful fact to quantum LDPC codes is that
regardless of the value $c$ the standard syndrome decoding can correct up to $\lfloor\frac{d-1}{2}\rfloor$ phase flips
and up to $\lfloor\frac{d-1}{2}\rfloor$ bit flips through two separate steps,
where each decoding step utilizes $H$ to compute the error syndrome for one of the two kinds of error \cite{HBD}.
In other words, we can take advantage of the ``classical minimum distance'' of ingredients in a straightforward manner
as we would in the classical setting during decoding.
In what follows, when the distance of a quantum error-correcting code is discussed,
we generally focus on this type of straightforwardly exploitable classical minimum distance of a quantum LDPC code.

The \textit{Tanner graph} of an $m \times n$ parity-check matrix $H$
is the bipartite graph consisting of $n$ bit vertices and $m$ parity-check vertices,
where an edge joins a bit vertex to a parity-check vertex
if that bit is included in the corresponding parity-check equation.
A \textit{cycle} in a graph is a sequence of connected vertices
which starts and ends at the same vertex in the graph and contains no other vertices more than once.
The \textit{girth} of a parity-check matrix is the length of a shortest cycle in the corresponding 
Tanner graph.
Since the Tanner graph is bipartite, its girth is even.
Clearly, a $4$-cycle in a parity-check matrix is a $2 \times 2$ all-one submatrix.
A $6$-cycle is a $3 \times 3$ submatrix in which each row and column has exactly two ones.
Typically, $4$-cycles severely reduce error correction performance while $6$-cycles have a 
mild negative effect.
Since we are interested in codes with excellent performance, we only consider codes with girth at least six.
To avoid triviality, we also assume that the row and column weights of a parity-check matrix are at least two.
An LDPC code is \textit{regular} if its parity-check matrix has constant row and column weights, and \textit{irregular} otherwise.

We begin with a simple observation about the the structure of a classical LDPC code without short cycles
which form a homogenous quantum LDPC code requiring only one ebit.

\begin{theorem}\label{observation}
There exists a homogeneous quantum \textup{LDPC} code which requires only one ebit and has girth greater than four if and only if
the following conditions on the parity-check matrix of the corresponding classical \textup{LDPC} code hold:
\begin{enumerate}
\item For each pair of distinct parity-checks $r_i$, $r_j$ there exists exactly one bit involved in both $r_i$ and $r_j$,
\item The size of each parity-check is odd and greater than one,
\item Each bit is involved in more than one parity-check.
\end{enumerate}
\end{theorem}
\begin{proof}
First we prove sufficiency.
Let $H$ be a parity-check matrix of a classical LDPC code
which satisfies the three conditions in the statement.
Since the sizes of parity-checks are odd, the entries on the diagonal of $HH^T$
representing the inner products of the same rows are ones.
Because for each pair of parity-checks there exists exactly one bit involved in the pair, the other entries of $HH^T$ are also ones.
Hence $HH^T$ is the all-one square matrix. Thus, we have $\rk{HH^T} = 1$.
Since no pair of rows of $H$ produces a $4$-cycle, the girth of the corresponding classical LDPC code is greater than four.

Next we prove necessity.
Let $H'$ be a parity-check matrix of a classical binary linear code which yields
a homogeneous quantum LDPC code of girth greater than four requiring only one ebit.
Let $H'H'^T = (h'_{i,j})$ and write the $i$th row and $j$th column of $H'H'^T$
as $\boldsymbol{r_i}$ and $\boldsymbol{c_j}$ respectively.
Since $H'H'^T$ is not a zero matrix, there exists a nonzero entry $h'_{a,b}$.
Let $I = \{i : h'_{i,b} = 1\}$.
Because $H'H'^T$ is symmetric with respect to the diagonal representing the inner products of the same rows corresponding to parity-checks,
we have $h'_{b,i} = 1$ for $i \in I$.
Since $\rk{H'H'^T} = 1$, the rows $\boldsymbol{r_i}$ and columns $\boldsymbol{c_i}$, $i \in I$,
induce the $|I| \times |I|$ all-one matrix in $H'H'^T$.
Since $H'$ has no $4$-cycles, the all-one matrix corresponds to a set $R$ of rows in $H'$,
where for any $\boldsymbol{r} \in R$ the weight of $\boldsymbol{r}$ is odd and
each pair of distinct rows in $R$
have exactly one position in which both entries are ones.
Hence, if $H' = R$, then $H'$ satisfies all the three conditions in the statement.
Suppose the contrary, that $H'$ has a row $\boldsymbol{r'} \not\in R$.
Then $\boldsymbol{r'}$ does not have a one in a position where $\boldsymbol{r} \in R$ does.
By assumption, every column of $H'$ has at least two ones, and hence there is another row $\boldsymbol{r''}$
which is not in $R$ and has a one in one of the positions in which $\boldsymbol{r'}$ has a one.
Since $\rk{H'H'^T} = 1$ and $R$ generates all-one submatrix of $H'H'^T$,
$\boldsymbol{r'}$ and $\boldsymbol{r''}$ are orthogonal.
Hence the pair of rows induce a $4$-cycle, a contradiction. This completes the proof.
\end{proof}

Combinatorial objects which are equivalent to the classical LDPC codes
satisfying the conditions in Theorem \ref{observation} have been studied since the late $19$th century in combinatorial design theory.
Let $K$ be a subset of positive integers.
A \textit{pairwise balanced design} of \textit{order} $v$ and \textit{index} one with \textit{block sizes} from $K$, denoted by PBD$(v,K,1)$,
is an ordered pair $(V, {\mathcal B})$, where $V$ is a finite set of $v$ elements
called \textit{points}, and ${\mathcal B}$ is a family of subsets
of $V$, called \textit{blocks}, that satisfies the following two conditions:
\begin{enumerate}
	\item[(i)] each unordered pair of distinct elements of $V$ is contained in exactly one block of ${\mathcal B}$, 
        \item[(ii)] for every $B \in {\mathcal B}$ the cardinality $|B| \in K$.
\end{enumerate}
When $K$ is a singleton $\{\mu\}$, the PBD is a \textit{Steiner} $2$-\textit{design} of \textit{order} $v$
and \textit{block size} $\mu$, denoted by $S(2,\mu,v)$.
A PBD of order $v$ is \textit{trivial} if it has no blocks or consists of only one block of size $v$.
When the cardinality of $V$ is positive, a trivial PBD with no blocks means that $V$ is a singleton.

Define $\alpha(K) = \gcd\{\mu-1 : \mu \in K\}$ and $\beta(K) = \gcd\{\mu(\mu-1) : \mu \in K\}$.
Necessary conditions for the existence of a PBD$(v,K,1)$ are $v-1 \equiv 0 \pmod{\alpha(K)}$
and $v(v-1) \equiv 0 \pmod{\beta(K)}$. By a constructive proof, these conditions were shown to be asymptotically sufficient:
\begin{theorem}[Wilson \cite{W2}]\label{th:existencePBD}
There exists a constant $v_K$ such that
for every $v > v_K$ satisfying $v-1 \equiv 0 \pmod{\alpha(K)}$ and $v(v-1) \equiv 0 \pmod{\beta(K)}$
there exists a \textup{PBD}$(v,K,1)$.
\end{theorem}

The \textit{replication number} $r_x$ of a point $x \in V$ of a PBD $(V, {\mathcal B})$
is the number of occurrences of $x$ in the blocks of ${\mathcal B}$.
A PBD is \textit{odd-replicate} if for every $x \in V$ the replication number $r_x$ is odd.
If $r_x = r_y$ for any two points $x$ and $y$,
we say that the PBD is \textit{equireplicate} (or \textit{regular}) and has replication number $r_x$.
While our result will show that regular PBDs give rise to LDPC codes that are right-regular in the language of coding theory,
to avoid any confusion, we use the term equireplicate for combinatorial designs.
Every $S(2,\mu,v)$ is equireplicate and has replication number $\frac{v-1}{\mu-1}$.
An \textit{incidence matrix} of a PBD $(V,{\mathcal B})$ is a binary $v \times b$ matrix
 $H = (h_{i,j})$ with 
rows indexed by points, columns indexed by blocks, and $h_{i,j}=1$ if the
 $i$th point is contained in the $j$th block, and $h_{i,j}=0$ otherwise.

\begin{theorem}\label{th:equivalencePBD}
There exists a homogeneous quantum \textup{LDPC} code which requires only one ebit and has girth greater than four if and only if
the corresponding parity-check matrix of the classical \textup{LDPC} code is an 
incidence matrix of a nontrivial odd-replicate \textup{PBD} with index one and smallest
block size greater than one.
\end{theorem}
\begin{proof}
Let $H$ be an incidence matrix of a nontrivial odd-replicate \textup{PBD} with index one and smallest block size greater than one.
It suffices to show that $H$ is equivalent to a parity-check matrix satisfying the conditions on the classical LDPC code in Theorem \ref{observation}.
Because every pair of points appear exactly once in a block, for every pair of rows there exists exactly one column where both rows have one.
The number of appearances of a point is the weight of the corresponding row in $H$, which is odd and not equal to one.
Because each block contains more than one point, the weight of each column is larger than one.
By indexing rows of $H$ by parity-checks and columns by bits, $H$ can be regarded as a parity-check matrix satisfying the conditions as required.
It is trivial that the converse also holds.
\end{proof}

Note that if we allow a column of weight one, without loss of generality,
the parity-check matrix $H$ of a classical LDPC code must be either an incidence matrix of an odd-replicate \textup{PBD} with index one or of the form
\[\left[\begin{array}{cc}
A & 0 \\
0 & B
\end{array}\right],\]
where $A$ is an incidence matrix of an odd-replicate \textup{PBD}
with index one and smallest block size greater than one,
and $B$ is of constant column weight one and satisfies $BB^T = 0$.
Hence, in the latter case, $H$ defines a classical code which is either of minimum distance two
or consists of codewords with zeros added to each codeword defined by the PBD.
Hence, we only consider the case when each row and column has at least two ones.

The necessary and sufficient condition given in Theorem \ref{th:equivalencePBD} allows us to prove
that homogeneous quantum LDPC codes requiring only one ebit must have girth less than or equal to six.
\begin{theorem}\label{th:nogirtheight}
There exists no homogeneous quantum \textup{LDPC} code with girth greater than six which requires only one ebit.
\end{theorem}
\begin{proof}
Suppose the contrary, that there exists a parity-check matrix $H$ of a binary linear code
which yields a homogeneous quantum LDPC code with girth greater than six requiring only one ebit.
By Theorem \ref{th:equivalencePBD}, the parity-check matrix $H$ is an incidence matrix of a PBD of index one.
Take an arbitrary column $\boldsymbol{c_1}$ of $H$.
Write the block $B_1$ which corresponds to $\boldsymbol{c_1}$ as $\{v_1, \dots, v_{|B_1|}\}$.
Since every row has at least two ones, we can find another column $\boldsymbol{c_2}$
which corresponds to $B_2 = \{v_1, v_{|B_1|+1},\dots, v_{|B_1|+|B_2|-1}\}$,
where $v_i \not= v_j$ for any $i$ and $j$, $i \not= j$.
Take the third column $\boldsymbol{c_3}$
representing the block $B_3$ which contains the pair $\{v_2, v_{{|B_1|}+1}\}$.
The three columns $\boldsymbol{c_1}$, $\boldsymbol{c_2}$, and $\boldsymbol{c_3}$ induce a $6$-cycle, a contradiction.
\end{proof}

Thus, a homogeneous quantum LDPC code requiring only one ebit has girth six,
which is the largest possible,
if and only if the code is obtained from an odd-replicate PBD of index one.

An important case is when the classical ingredient is a regular LDPC code.
In this case, a simple necessary condition is asymptotically sufficient:
\begin{theorem}\label{th:existenceregularEAQECC}
A necessary condition for the existence of a regular homogeneous quantum \textup{LDPC} code which requires only one ebit and
is of length $n$, girth six, and column weight $\mu$ is that the number
$\frac{-1+\sqrt{1+4n\mu(\mu-1)}}{2(\mu-1)}$ is an odd integer.
Conversely, for any integer $\mu \geq 2$ there exists a constant $n_{\mu}$
such that for $n > n_{\mu}$ the necessary condition is sufficient.
\end{theorem}
\begin{proof}
Let $H$ be a parity-check matrix of a classical LDPC code.
Assume that $H$ yields a regular homogeneous quantum LDPC code
which requires only one ebit and is of length $n$, girth six, and column weight $\mu$.
By Theorem \ref{th:equivalencePBD}, $H$ forms an incidence matrix of a PBD of index one.
Because $H$ can also be seen as a parity-check matrix of a classical regular LDPC code,
the column weights are uniform. Hence, $H$ can be viewed as an incidence matrix of an $S(2,\mu,v)$ for some $v$.
Because the number of blocks of an $S(2,\mu,v)$ is $\frac{v(v-1)}{\mu(\mu-1)}$, we have $n = \frac{v(v-1)}{\mu(\mu-1)}$.
Hence,
\begin{eqnarray}
v = \frac{1+\sqrt{1+4n\mu(\mu-1)}}{2}.\label{eq1}
\end{eqnarray}
The number of occurrences of each point of an $S(2,\mu,v)$ is $\frac{v-1}{\mu-1}$.
Since $H$ defines an odd-replicate design, a necessary condition for the existence of a homogeneous quantum LDPC code
satisfying the stated properties is that \[\frac{v-1}{\mu-1} = \frac{-1+\sqrt{1+4n\mu(\mu-1)}}{2(\mu-1)}\] is odd.
Assume that the necessary condition holds.
Then, we have \[\frac{1+\sqrt{1+4n\mu(\mu-1)}}{2}-1 \equiv 0 \pmod{\mu -1}\]
and
\begin{eqnarray*}
\lefteqn{\frac{1+\sqrt{1+4n\mu(\mu-1)}}{2}(\frac{1+\sqrt{1+4n\mu(\mu-1)}}{2}-1)}\hspace{45mm}\\
&=& n\mu(\mu-1)\\
&\equiv& 0 \pmod{\mu(\mu-1)}.
\end{eqnarray*}
Applying Theorem \ref{th:existencePBD} by plugging $\alpha(K) = \mu -1$ and $\beta(K) = \mu(\mu-1)$ completes the proof.
\end{proof}

As we have seen in this section, there is a strong relation between homogeneous quantum LDPC codes and Steiner $2$-designs.
This equivalence implies that the framework given in \cite{FCVBT} encompasses
all regular homogeneous quantum LDPC codes with girth six which require only one ebit.

Particularly useful facts are that the original proof of Theorem \ref{th:existencePBD} is constructive
and that there are many known explicit constructions for PBDs with various properties.
For more details on explicit combinatorial constructions useful to entanglement-assisted quantum LDPC codes,
we refer the reader to \cite{FCVBT} and references therein.

\section{Rates, distances, and numbers of $6$-cycles}\label{sc:Rates}
Next we examine the possible code parameters.
As shown in the proof of Theorem \ref{th:equivalencePBD}, the number of blocks of a PBD$(v, K, 1)$ corresponds to the code length.
The number of rows of the parity-check matrix of the underlying classical LDPC code is the number $v$ of points.
The number of points in each block is the weight of the corresponding column in the parity-check matrix.
Hence, $K$ determines the possible column weights.
Because the parity-check equations are labeled by the points of the PBD,
the weight of each row is the replication number of the corresponding point.
If the corresponding classical LDPC code is regular,
its parity-check matrix forms an incidence matrix of an $S(2,\mu,v)$,
which means that the code is of length $\frac{v(v-1)}{\mu(\mu-1)}$, constant column weight $\mu$, and constant row weight $\frac{v-1}{\mu-1}$.
The number of rows is $v$.
In the reminder of this section, we investigate code parameters further in detail.

We first consider the rates for the case when classical ingredients are regular.
The dimension of a homogeneous quantum LDPC code is determined by the rank of the corresponding parity-check matrix of the classical LDPC code.
By Theorem \ref{th:equivalencePBD}, we only need to know the rank of the incidence matrix of the combinatorial design equivalent to the classical code.
Hillebrandt \cite{hillebrandt} gave a bound on the rank of an incidence matrix of a Steiner $2$-design.
\begin{theorem}[Hillebrandt \cite{hillebrandt}]\label{bound:rkSteiner}
The rank of an incidence matrix $H$
of an $S(2,\mu,v)$ satisfies the following  inequalities:
\[ \left\lceil \frac{1}{2}+\sqrt{\frac{1}{4}+\frac{(v-1)(v-\mu)}{\mu}}\right\rceil \leq \rk{H} \leq v.\]
\end{theorem}
Hence, we have the following bound on the dimension:
\begin{theorem}\label{th:regulardimensionbound}
If there exists a regular homogeneous quantum \textup{LDPC} code with girth six and column weight $\mu$
whose parameters are
$[[n,k;1]]$, then
\[n - \sqrt{1+4n\mu(\mu-1)} \leq k\]
and
\[k \leq  n - 2\left\lceil \frac{1}{2}+\sqrt{\frac{1}{4}+\frac{(v-1)(v-\mu)}{\mu}}\right\rceil + 1,\]
where
\[v = \frac{1+\sqrt{1+4n\mu(\mu-1)}}{2}.\]
\end{theorem}
\begin{proof}
As is stated in Equation (\ref{eq1}),
a regular homogeneous quantum LDPC code with length $n$, girth six, and column weight $\mu$ which requires only one ebit
must be constructed from an incidence matrix $H$ of a Steiner $2$-design
of order $\frac{1+\sqrt{1+4n\mu(\mu-1)}}{2}$ and block size $\mu$.
Since the incidence matrix requires only one ebit,
the dimension of the quantum LDPC code is $k = n - 2\rk{H} +1$.
Applying Theorem \ref{bound:rkSteiner} to this relation between the dimension $k$ and the rank of $H$ completes the proof.
\end{proof}

If one wishes to obtain a code of highest possible rate for a given length and row and column weights, 
the incidence matrix of the corresponding Steiner 2-design must be of minimum rank.
A Steiner 2-design $S(2,3,2^m -1)$ has odd replication number equal to $2^{m-1} -1$.
It is known that the rank of any $S(2,3,2^m -1)$ is greater than or equal to $2^m -1 -m$,
and the minimum $2^m - 1 -m$ is achieved if and only if the design is isomorphic to the classical
design whose points and blocks are the points and lines of the binary
projective geometry PG$(m-1,2)$ \cite{DHV}.

An odd-replicate $S(2,3,v)$ exists if and only if $v \equiv 3, 7 \pmod{12}$ \cite{BJL}.
The ranks of such designs were determined by Assmus \cite{A}.

\begin{theorem}[Assmus \cite{A}]\label{th:assmus}
For any $v \equiv 3, 7 \pmod{12}$, where $v = 2^tu-1$ and $u$ is odd,
and any integer $i$ with $1 \leq i < t$,
there exists an $S(2,3,v)$ of rank equal to $v-t+i$.
\end{theorem}

As a corollary, we have the following necessary and sufficient conditions.

\begin{theorem}\label{th:STSEAQECCdimension}
Let $n > 7$ be an integer.
There exists a regular homogeneous quantum \textup{LDPC} code of length $n$, dimension $k$, girth six, and column weight three
which requires only one ebit if and only if
\[\sqrt{24n+1} \equiv 5 \pmod{8}\]
and
\[n - \sqrt{24n+1} \leq k \leq n - \sqrt{24n+1} + 2t - 2,\]
where $t$ is the integer satisfying $\sqrt{24n+1} = 2^{t+1}u-3$ with $u$ odd.
\end{theorem}
\begin{proof}
For every $v \equiv 1, 3 \pmod{6}$, $v > 7$, there exists an $S(2,3,v)$ of full rank \cite{DHV}. 
Theorem \ref{th:assmus} provides all possible $S(2,3,v)$s with deficient ranks.
\end{proof}

It is notable that Theorems \ref{th:regulardimensionbound} and \ref{th:STSEAQECCdimension} suggest that
homogeneous quantum LDPC codes requiring only one ebit typically have very high rates.
In fact, because the number of columns in an incidence matrix of a Steiner $2$-design
is the largest possible for a matrix with a given number of rows that avoids $4$-cycles,
the rate of the corresponding classical LDPC code is the highest possible in a sense.

At the same time, however, the extremely high rates imply that it is impossible to obtain an infinite family of
LDPC codes of rate bounded away from one for some reasonable degree distribution.
Because the number of blocks in a PBD of order $v$ is $c\cdot v^2$ for some constant $c$,
the code length is $c\cdot v^2$ with $v$ being the number of rows of the corresponding parity-check matrix.
Hence, the combinatorial design theoretic construction described here is more suitable
when the code designer wishes to deterministically design a code of moderate length
with specific properties desirable for a particular purpose.

For instance, it is known that redundant rows in a parity-check matrix can help improve error correction performance of the sum-product algorithm \cite{Mbook}.
One might then wish to design a parity-check matrix with a large number of redundant rows while completely avoiding $4$-cycles.
The minimum distance should not be too small either in a normal situation.
Such a highly structured matrix $H$ would be nearly impossible
to obtain by a random draw when there is another stringent condition that $c = \rk{HH^T}$ must be kept small.
However, these conditions can easily be translated into the language of combinatorial designs,
and hence one might be able to tell whether such $H$ exists and, if it does, how to explicitly construct it.
In fact, the above constraints were effectively exploited to demonstrate
that high performance EAQECCs do not necessarily require a lot of ebits \cite{HYH2}.

The minimum distance of a binary linear code whose parity-check matrix forms an incidence matrix
of an $S(2,\mu,v)$ is at least $\mu + 1$
\footnote{This can be easily seen by taking an arbitrary block $B$ and
counting how many blocks it requires to form a linearly dependent set of columns in the corresponding parity-check matrix.
Because no pair of points appear in more than one block,
each additional column can share a one at at most one row with the column corresponding to $B$.
Hence, any linearly dependent set of columns in the parity-check matrix is of size at least  $\mu + 1$.}.
While it appears to be difficult to obtain the exact upper bound on the minimum distance in general,
incidence matrices of $S(2,\mu,v)$s can give minimum distances large enough
for the standard sum-product algorithm at moderate length (see \cite{FCVBT, HYH2}).
In fact, the Desarguesian projective plane of order $2^t$ gives an entanglement-assisted quantum LDPC code of
length $4^t+2^t+1$ and dimension $4^t+2^t-2\times3^t$ with the corresponding parity-check matrix being of minimum distance $2^t+2$,
which performs very well over the depolarizing channel.

To further study the minimum distances of LDPC codes based on Steiner $2$-designs,
we define combinatorial design theoretic notions.
A \textit{configuration} ${\mathcal C}$ in an $S(2,\mu,v)$, $(V,{\mathcal B})$,
is a subset ${\mathcal C} \subseteq {\mathcal B}$.
The set of points appearing in at least one block of a configuration ${\mathcal C}$
is denoted by $V({\mathcal C})$.
Two configurations ${\mathcal C}$ and ${\mathcal C}'$ are \textit{isomorphic}
if there exists a bijection $\phi : V({\mathcal C}) \rightarrow V({\mathcal C}')$
such that for each block $B \in {\mathcal C}$,
the image $\phi(B)$ is a block in ${\mathcal C}'$.
When $|{\mathcal C}|=i$,
a configuration ${\mathcal C}$ is  an \textit{$i$-configuration}.
A configuration ${\mathcal C}$ is  \textit{even}
if for every point $a$ appearing in ${\mathcal C}$
the number $|\{B: a \in B \in {\mathcal C}\}|$
of blocks containing $a$ is even.

The notion of minimum distance can be translated into the language of combinatorial designs.
An $S(2,\mu,v)$ is  {\it $r$-even-free} if for every integer $i$ satisfying $1\leq i \leq r$ it contains no even $i$-configurations.
Because the minimum distance of a binary linear code is the size of a smallest linearly dependent set of columns in its parity-check matrix,
the minimum distance of a linear code based on a Steiner $2$-design is determined by its even-freeness:
\begin{proposition}
The minimum distance of a binary linear code whose parity-check matrix forms an incidence matrix of a Steiner $2$-design is $d$
if and only if the corresponding Steiner $2$-design is $(d-1)$-even-free but not $d$-even-free.
\end{proposition}

A fairly tight bound on the minimum distance is available for the special case
when the parity-check matrix has constant column weight three and gives a regular LDPC code.

By definition every $r$-even-free $S(2,3,v)$, $r\geq 2$, is also $(r-1)$-even-free.
Every $S(2,3,v)$ is trivially $3$-even-free.
For $v >3$ an $S(2,3,v)$ may or may not be $4$-even-free.
Up to isomorphism, the only even $4$-configuration is the {\it Pasch} configuration.
It can be written by six points and four blocks:
$\{\{a, b, c\}, \{a, d, e\}, \{f, b, d\}, \{f, c, e\}\}$.
For the list of all the small configurations in an $S(2,3,v)$ and more complete treatments, we refer the reader to \cite{TRIPLESYSTEMS} and \cite{stopsts}.
Because every block in an $S(2,3,v)$ has three points, no $i$-configuration for $i$ odd  is even.
Hence, a $4$-even-free $S(2,3,v)$ is $5$-even-free as well, which means that an $S(2,3,v)$ is $5$-even-free if and only if it contains no Pasch configuration.

The minimum distance $d$ of a classical LDPC code is the smallest number of columns in its parity-check matrix $H$
that add up to the zero vector over $\mathbb{F}_2^v$.
If $H$ forms an incidence matrix of an $S(2,3,v)$, then
a set of $d$ columns that add up to the zero vector is equivalent to an even $d$-configuration in the $S(2,3,v)$.
Hence, as an immediate corollary of Theorem \ref{th:equivalencePBD} and Equation (\ref{eq1}), we have the following proposition:

\begin{proposition}\label{evenisdistance}
A classical LDPC code of length $n$, minimum distance $d$, and constant column weight three
forms a homogeneous quantum LDPC code of girth greater than four requiring only one ebit
if and only if $\frac{-1+\sqrt{1+24n}}{4}$ is an odd integer and
the parity-check matrix forms an incidence matrix of a $(d-1)$-even-free $S(2,3,\frac{1+\sqrt{1+24n}}{2})$ that is not $d$-even-free.
\end{proposition}

As far as the authors are aware, the sharpest known upper bound on the even-freeness of an $S(2,3,v)$ is the one found in the study of X-tolerant circuits:
\begin{theorem}[Fujiwara and Colbourn \cite{FC}]\label{no8evenfreeSTS}
For $v>3$ there exists no nontrivial $8$-even-free $S(2,3,v)$.
\end{theorem}

Hence, by Proposition \ref{evenisdistance}, Theorem \ref{no8evenfreeSTS}, and the fact that every $S(2,3,v)$ is $3$-even-free,
we obtain a bound on the minimum distance of the classical ingredient:
\begin{theorem}\label{th:paraSTSEAQECC}
If there exists a regular homogeneous quantum \textup{LDPC} code with girth six and column weight three requiring only one ebit,
then the minimum distance $d$ of the corresponding classical LDPC code satisfies $4 \leq d \leq 8$.
\end{theorem}

As is the case with $S(2,3,v)$s, in general, odd-replicate $(d-1)$-even-free $S(2,\mu,v)$s that are not $d$-even-free are equivalent to
classical regular LDPC codes of constant column weight $\mu$, girth six, and minimum distance $d$
that generate homogeneous quantum LDPC codes requiring only one ebit.
However, there do not seem to exist many results on the even-freeness of $S(2,\mu,v)$s
or equivalently the minimum distances of the corresponding classical regular LDPC codes in the literature.
To the best of the authors' knowledge,
the following explicit construction gives the highest known even-freeness for $\mu \geq 3$:

\begin{theorem}[M\"{u}ller and Jimbo \cite{JIMBOERASURE}]\label{affineerasure}
For any odd prime power $q$ and positive integer $m \geq 2$
the points and lines of affine geometry $AG(m,q)$ form a $(2q-1)$-even-free $S(2,q,q^m)$.
\end{theorem}
When $m$ is odd, the $S(2,q,q^m)$ is odd-replicate, and the size of each parity-check is $\frac{q^m}{2}$.
Hence, we obtain quantum LDPC codes requiring only one ebit in this case.
For a more detailed treatment of explicit constructions and the performance of quantum LDPC codes obtained from finite geometry,
we refer the reader to \cite{FCVBT, HYH2}.

When an LDPC code is decoded by the standard sum-product algorithm,
$6$-cycles may affect error correction performance of an LDPC code in a negative manner.
It is known that the number $N_6$ of $6$-cycles in a parity-check matrix from an incidence matrix
of an $S(2,\mu,v)$ is exactly $\frac{v(v-1)(v-\mu)}{6}$ (see, for example, Johnson and Weller \cite{JW2}).
Hence, by Equation (\ref{eq1}), if we decode a regular homogeneous quantum regular LDPC code
which requires only one ebit and is of length $n$, girth six, and column weight $\mu$ in two separate steps
by using the same parity-check matrix for both X and Z errors, each step involves
\[N_6 = \frac{n\mu(\mu-1)(1-2\mu+\sqrt{1+4n\mu(\mu-1)})}{12}\]
6-cycles.

\section{Conclusion}\label{sc:Conclusion}
We have shown that homogeneous quantum LDPC codes requiring only one ebit and avoiding $4$-cycles
are equivalent to special classes of fundamental combinatorial designs.
Various properties of entanglement-assisted quantum LDPC codes
have been revealed by applying known theorems and techniques of combinatorial design theory.
Our results also give theoretical insight into the known entanglement-assisted quantum LDPC codes
presented as counterexamples to the conjecture on the required amount of entanglement.

We have demonstrated that combinatorial design theory may work as a useful mathematical tool
to investigate homogeneous quantum LDPC codes consuming one ebit.
It  will be interesting to study the case when two or more ebits are allowed and investigate how much information we can extract
about the structure of such quantum LDPC codes.

Another important direction would be to investigate quantum LDPC codes obtained from pairs of distinct classical LDPC codes.
As we have seen in Section \ref{sc:Rates}, if we only allow one ebit,
the rate of a homogeneous quantum LDPC code of girth six approaches one as the length becomes larger.
Hence, it would be quite interesting to investigate whether heterogeneous quantum LDPC codes can overcome this fundamental limitation
while avoiding short cycles and suppressing the number of required ebits.
Another possible merit of studying the heterogenous case from the viewpoint of combinatorics would be that
combinatorial methods appear to be helpful to design highly structured quantum LDPC codes.
One possible direction would be to study how to optimize codes for an asymmetrical quantum channel
where the probabilities of bit flips and phase flips are not equal (see \cite{IM,SKR,FH}).
We hope that these questions will be answered in future work.

\section*{Acknowledgments}
The authors would like to thank Mark M. Wilde, the anonymous referee and Associate Editor Jean-Pierre Tillich for their insightful comments and valuable suggestions.
This research was conducted when the first author was visiting the Department of Mathematical Sciences, Michigan Technological University.
He thanks the department for its hospitality.


\begin{IEEEbiographynophoto}{Yuichiro Fujiwara}
(M'10) received the B.S. and M.S. degrees in mathematics from Keio University, Japan,
and the Ph.D. degree in information science from Nagoya University, Japan.

He was a JSPS postdoctoral research fellow with the Graduate School of
System and Information Engineering, Tsukuba University, Japan, and a visiting
scholar with the Department of Mathematical Sciences, Michigan Technological University.
He is currently with the Division of Physics, Mathematics and Astronomy,
California Institute of Technology, Pasadena, where he works as a visiting postdoctoral research fellow.

Dr.\ Fujiwara's research interests include combinatorics and its interaction with computer science and quantum information science,
with particular emphasis on combinatorial design theory, algebraic coding theory, and quantum information theory.
\end{IEEEbiographynophoto}

\begin{IEEEbiographynophoto}{Vladimir D. Tonchev}
graduated with PhD in Mathematics from the University of Sofia,
Bulgaria, in 1980, and received the  Dr. of Mathematical Sciences
degree from the Bulgarian Academy of Sciences in 1987.
After spending a year as a research fellow at the Eindhoven University
of Technology, The Netherlands, (1987-88), and two years at the
universities of Munich, Heidelberg and Giessen in Germany as
an Alexander von Humboldt Research Fellow (1988-90), Dr. Tonchev joined
Michigan Technological University, where he is currently
a Professor of Mathematical Sciences.
Tonchev has published over 160 papers, four books,
three book chapters, and edited several volumes in the area of
error-correcting codes,
combinatorial  designs, and their applications.
Dr. Tonchev is a member of the editorial board of \textit{Designs, Codes and Cryptography},
\textit{Journal of Combinatorial Designs}, \textit{Applications and Applied Mathematics},
and \textit{Albanian Journal of Mathematics}, and a Founding Fellow of
the Institute of
Combinatorics and its Applications.
\end{IEEEbiographynophoto}

\end{document}